\documentclass{llncs}
\usepackage{url}
\usepackage{amsmath,amssymb}
\usepackage{xspace}
\usepackage[noend]{algorithm}
\usepackage[noend]{algorithmic}
\usepackage{tikz}
\usepackage{pgfplots}

\newcommand{\mr}{\textsc{MapReduce}\xspace}
\newcommand{\dmr}{\textsc{DS-MapReduce}\xspace}
\newcommand{\mrs}{\textsc{MapReduce$(\sigma)$}\xspace}
\newcommand{\algomr}{\textsc{Algorithm} $\mathcal{MR}(\alpha,\gamma)$\xspace}
\newcommand{\fcfs}{\textsc{FCFS}\xspace}
\newcommand{\mrf}{\textsc{MapReduce$(\fcfs)$}\xspace}
\newcommand{\mrsr}{\textsc{MapReduce$(\sr)$}\xspace}
\newcommand{\sr}{\textsc{SR}\xspace}

\begin{document}

\title{Energy Efficient Scheduling of MapReduce Jobs
}

\author{Evripidis Bampis\inst{1}
\and Vincent Chau\inst{2}
\and Dimitrios Letsios\inst{1}
\and Giorgio Lucarelli\inst{1}
\and Ioannis Milis\inst{3}
\and Georgios Zois\inst{1,3}
}

\institute{Sorbonne Universit\'es, UPMC Univ Paris 06, UMR 7606, LIP6, F-75005, France.\\
\email{\{Evripidis.Bampis,Dimitrios.Letsios,Giorgio.Lucarelli, Georgios.Zois\}@lip6.fr}
\and IBISC, Universit\'{e} d'\'{E}vry, France.\\
\email{vincent.chau@ibisc.univ-evry.fr}
\and Dept. of Informatics, AUEB, Athens, Greece.\\
\email{milis@aueb.gr}
}

\maketitle

\begin{abstract}
MapReduce is emerged as a prominent programming model for data-intensive computation.
In this work, we study power-aware MapReduce scheduling in the speed scaling setting first introduced by Yao et al.~[FOCS 1995].
We focus on the minimization of the total weighted completion time of a set of MapReduce jobs under a given budget of energy.
Using a linear programming relaxation of our problem, we derive a polynomial time constant-factor approximation algorithm.
We also propose a convex programming formulation that we combine with standard list scheduling policies,
and we evaluate their performance using simulations.
\end{abstract}

\section{Introduction}

MapReduce has been established
as a standard programming model for parallel computing in data centers or computational grids and it is currently used
for several applications including search indexing, web analytics  or data mining.
However, data centers consume an enormous amount of energy and hence, energy efficiency has emerged as an important issue in the data-processing framework.
Several empirical works have been carried-out in order to study different mechanisms for the reduction of the energy consumption in the MapReduce setting
and especially for the Hadoop framework \cite{FellerRM13,FengLZY12,GoiriLNGTB12}.
The main  mechanisms for energy saving are the \emph{power-down} mechanism where in periods of low-utilization some servers are switched-off,
and the \emph{speed-scaling} mechanism (or DVFS for Dynamic Voltage Frequency Scaling) where the servers' speeds may be adjusted dynamically \cite{YaoDS95}.
Until lately, most work in the MapReduce framework  were focused on the  power-down mechanism, but recently,
Wirtz and Ge \cite{WirtzG11} showed that for some computation intensive MapReduce applications the use of intelligent speed-scaling
may lead to significant energy savings.
In this paper, we study power-aware MapReduce scheduling in the speed scaling setting from a theoretical point of view.

In a typical MapReduce framework, the execution of a MapReduce job creates a number of Map and Reduce tasks.
Each Map task processes a portion of the input data and outputs a number of key-value pairs.
All key-value pairs having the same key are then given to a Reduce task which processes the values associated with a key to generate the final result.
This means that each Reduce task cannot start before the completion of the last Map task of the same job.
In other words, there is a complete bipartite graph implying the precedences between Map and Reduce tasks of a job.
However, the Map tasks of a job can be executed in parallel and the same holds for its Reduce tasks.

In what follows we consider a set of MapReduce jobs that have to be  executed on a set of speed-scalable processors,
i.e., on processors that can adjust dynamically their speed \cite{YaoDS95}.
In the speed scaling setting, each task is associated with a work volume instead of a processing time
and the scheduler  has to decide not only the processor and the time interval where a task is executed,
but also its speed over time, taking into account the energy consumption.
High processor's speeds are in favor of performance  at the price of high energy consumption.
Each job consists of a set of Map tasks and a set of Reduce tasks, with every task having a positive work volume.
Each job is also associated with a positive weight representing its importance/priority, and a release date (or arrival time).
Like in~\cite{ChangKKLLM11,ChenKL12}, we consider that the Map and the Reduce tasks of each job are \emph{preassigned} to the processors and in this way we take into account data locality, i.e. the fact that each Map task has to be executed on the server where
its data are located. Given that  the preemption of tasks, i.e. the possibility of interrupting a task and resuming it later, may cause important overheads we do not allow it. This is also the case often in practice: Hadoop does not offer the possibility of preemption \cite{MoseleyDKS11}.
Our goal is to schedule all the tasks to the processors,
so as to minimize the total weighted completion time of jobs respecting a given budget of energy.

\paragraph{Related Work.}
Chang et al.~\cite{ChangKKLLM11} consider a set of MapReduce jobs with their Map and Reduce tasks
preassigned to processors and their goal is to minimize the total weighted completion time of jobs.
They proposed approximation algorithms of ratios 3 and 2 for arbitrary and common release dates, respectively.
However, they do not consider neither distinction nor dependencies between Map and Reduce tasks of a job.
Moreover, their model  falls  into a well-studied problem known as  {\em concurrent open-shop}
(or {\em order scheduling}) for which the same approximation results are known  (see ~\cite{MastrolilliQSSU10} and the references therein).
Extending on the above-mentioned model, Chen et al.~\cite{ChenKL12},
proposed a more realistic one which takes into account the dependencies among Map and Reduce tasks  and derived  an 8-approximation algorithm for the same objective.
Moreover, they managed to model  also the  transfer of the output of Map tasks to Reduce tasks and to derive  a  58-approximation algorithm for this generalization.
In a third model proposed by Moseley et al.~\cite{MoseleyDKS11}, the dependencies between Map and Reduce tasks of a job are also  taken into account while the assignment of tasks to processors is not given in advance.
The authors studied the preemptive variant for both the case of identical and unrelated processors. They proposed constant approximation ratios of 12 and 6, respectively.
For the unrelated processors case, they focused on the special case where each job has a single Map and a single Reduce task.
For the latter case on a single map and a single reduce processor they also proposed a QPTAS which becomes a PTAS for a fixed number of processing times of tasks. Recently, in~\cite{FotakisMZZ13} the authors proposed a $32+\epsilon$-approximation algorithm for the unrelated processors case with multiple Map and Reduce tasks per job.

In the energy-aware setting, Angel et al.~\cite{AngelBK12}
proposed approximation algorithms for the problem of minimizing
the total weighted completion time on unrelated parallel
processors, under a model where the processing time and the energy
consumption of the jobs are speed dependent. Other works in this setting,
related to our problem, deal with single processor problems.
Megow et al. ~\cite{MegowV12} recently proposed a PTAS for the problem of minimizing the total weighted completion time on a single speed-scalable processor.

\paragraph{Our Results and Organization of the Paper}
We adopt the MapReduce model of~\cite{ChangKKLLM11} where the tasks are preassigned to processors
but extended with dependencies between Map and Reduce tasks
as in Chen et al.~\cite{ChenKL12,MoseleyDKS11} in the speed scaling setting \cite{YaoDS95}.
After a formal statement of our problem and notation, we present, in Section~\ref{se:lpa}, a polynomial time LP-based $O(1)$-energy $O(1)$-approximation algorithm which allows energy augmentation, i.e., it may use more energy than an optimal solution which always respects the energy budget, while using discretization of the possible speed values and list scheduling in the order of tasks' $\alpha$-points (see e.g. \cite{PhillipsSW97,HallSW96,Skutella06}).
As we show, there is a tradeoff between the approximation ratio and energy augmentation as a function of $\alpha$, where the schedule is converted to a constant-factor approximation for our problem.
In Section~\ref{se:cp}, we are interested in natural list scheduling policies such as \textsc{First Come First Serve} (\fcfs) and \textsc{Smith Rule} (\sr). However, in our context we need to determine the speeds of every task in order to respect the energy budget. For that, we propose a convex programming relaxation of our problem when an order of the jobs is prespecified.
This relaxation can be solved in polynomial time to arbitrary precision by the Ellipsoid algorithm \cite{NesterovNN94}.
Then we combine the solution of this relaxation with  \fcfs and \sr and  we compare experimentally their effectiveness.
Finally, we conclude in Section~\ref{se:con}.

\section{Problem Definition and Notation}

In the sequel we consider a set $\mathcal{J}=\{1,2,\ldots,n\}$ of $n$ MapReduce jobs to be executed
on a set $\mathcal{P}=\{1,2,\ldots,m\}$ of $m$ speed-scalable processors.
Each job is associated with a positive weight $w_j$ and a release date $r_j$
and consists of a set of Map tasks and a set of Reduce tasks that are preassigned to the $m$ processors.
We denote by $\mathcal{T}$ the set of all tasks of all jobs,
and by $\mathcal{M}$ and $\mathcal{R}$ the sets of all Map and Reduce tasks, respectively.
Each task $T_{i,j} \in \mathcal{T}$ is associated a non-negative work volume $v_{i,j}$.

We consider each job having at least one Map and one Reduce task and
that each job  has at most one task, either Map or Reduce, assigned to each processor.
Map or Reduce tasks can run simultaneously on different processors, while the following precedence constraints hold for each job:
every Reduce task can start its execution after the completion of all Map tasks of the same job.

For a given schedule we denote by  $C_j$ and $C_{i,j}$ the completion times of each
job $j \in \mathcal{J}$ and each  task $T_{i,j} \in \mathcal{T}$, respectively.
Note that, due to the precedence constraints of Map and Reduce tasks, $C_j = \max_{T_{i,j}\in\mathcal{R}} \{C_{i,j}\}$.
By $C_{max}=\max_{j \in \mathcal{J}} \{C_j\}$ we denote the makespan of the schedule, i.e., the completion time of the job which finishes last.
Let also, $w_{\min}=\min_{j \in \mathcal{J}} \{w_j\}$, $v_{\min}=\min_{T_{i,j} \in \mathcal{T}} \{v_{i,j}: v_{i,j}>0\}$,
$w_{\max}=\max_{j \in \mathcal{J}} \{w_j\}$, $r_{\max}=\max_{j \in \mathcal{J}} \{r_j\}$ and $v_{\max}=\max_{T_{i,j} \in \mathcal{T}} \{v_{i,j}\}$.

In this paper, we combine this abstract model for MapReduce scheduling with the speed scaling mechanism for energy saving \cite{YaoDS95}
(see also \cite{Albers11} for a recent review).
In this setting the power required by a processor running at time $t$ with speed  $s(t)$
is equal to $P((s(t))=s(t)^\beta$, for a constant $\beta>1$ (typical values of $\beta$ are between 2 and 3)
and its energy consumption  is power integrated over  time, i.e., $E=\int P(s(t)) dt$.

Due to the convexity of the speed-to-power function,
a key property of our problem is that each task runs at a constant speed during its whole execution.
So, if a task $T_{i,j}$ is executed at a speed $s_{i,j}$, the time needed for
its execution (processing time) is equal to
$p_{i,j}=\frac{v_{i,j}}{s_{i,j}}$ and its energy consumption is
$E_{i,j}= \frac{v_{i,j}}{s_{i,j}}s_{i,j}^{\beta}=v_{i,j}s_{i,j}^{\beta-1}$.

Moreover, we are given an energy budget $E$ and the goal is to schedule \emph{non-preemptively} all the tasks to the $m$ processors,
so as to minimize the total weighted completion time of the schedule, i.e., $\sum_{j\in\mathcal{J}}w_j C_j$, without exceeding the energy budget $E$.
We refer to this problem as \mr problem.

As already mentioned above, the special case of the \mr problem where there are not dependencies between Map and Reduce tasks of each job
and each processor runs at a constant speed, reduces to the \emph{concurrent open-shop} problem which is known to be strongly $\mathcal{NP}$-complete~\cite{Roemer06}.
It is also easy to adapt $\mathcal{NP}$-completeness reductions like the one for the concurrent open-shop problem in the speed scaling setting,
and therefore, the \mr problem is also strongly $\mathcal{NP}$-hard.

\section{A Linear Programming Approach}\label{se:lpa}

In this section we present a constant-factor approximation algorithm for the \mr problem. Our algorithm allows energy augmentation and derives a $O(1)$-energy $O(1)$-approximation schedule for the problem. As we show, there is a tradeoff where the schedule can be converted to a constant-factor approximate schedule for the \mr problem.

Our algorithm is based on a formulation of the problem as a linear programming relaxation.
Then, we transform the solution obtained by the linear program to a feasible schedule for the \mr problem using the technique of $\alpha$-points.

\subsection{Discretization of Speeds}

Before presenting the linear programming formulation,
our first step is to discretize the possible speed values by loosing a factor of $(1+\epsilon)$ with respect to an optimal solution.
In order to do this, we need the following propositions that bound the length of an optimal schedule and the possible speed values.

\begin{proposition}\label{prop:T}
The makespan of any optimal schedule for the \mr problem is at most
\begin{equation*}
t_{\max} = \frac{w_{\max}}{{w_{\min}}}\left(nr_{\max} + n(n+1) \left(\frac{|\mathcal{T}| \cdot v_{\max}^{\beta}}{E}\right)^{\frac{1}{\beta-1}}\right)
\end{equation*}
\end{proposition}
\begin{proof}
Consider an optimal schedule for the \mr problem.
By definition, we have that $C_{\max}=\max_{j \in \mathcal{J}}\{C_j\}$.
Hence, it holds that $w_{\min} C_{\max} \leq \sum_{j \in \mathcal{J}} w_j C_j$.

In order to give an upper bound to $\sum_{j \in \mathcal{J}} w_j C_j$,
consider an instance of our problem where the weight $w_j$ and the release date $r_j$
of each job $j \in \mathcal{J}$ are rounded up to $w_{\max}$ and $r_{\max}$, respectively.
Moreover, assume that in this instance all tasks have work equal to $v_{\max}$.

Consider now an arbitrary order $\{1,2,\ldots,n\}$ of the jobs.
We create a feasible schedule $S$ for the modified instance as follows.
All tasks run with the same speed $s=\left(\frac{E}{|\mathcal{T}|\cdot v_{\max}}\right)^{1/(\beta-1)}$,
hence each task has a processing time $p=\frac{v_{\max}}{s}$.
Note that this speed allows us to execute all tasks without exceeding the energy budget.
As all tasks have the same processing time, we can consider the time horizon partitioned into time slots of length $p$ starting from $r_{\max}$.
For each job $j$, $1 \leq j \leq n$, we execute its Map tasks at time $r_{\max}+(2j-2)p$
and its Reduce tasks at time $r_{\max}+(2j-1)p$.
Then, for the objective value $\sum_{j \in \mathcal{J}} w_{\max} C_j^S$ of this schedule it holds that
\begin{eqnarray*}
\sum_{j \in \mathcal{J}} w_{\max} C_j^S
& = & w_{\max} \sum_{j=1}^n (r_{\max} + 2jp)\\
& = & w_{\max}\left(n r_{\max} + n(n+1)\frac{v_{\max}}{s}\right)
\end{eqnarray*}
The objective value of schedule $S$ is clearly an upper bound on the objective value
$\sum_{j \in \mathcal{J}} w_j C_j$ of an optimal schedule for the initial instance and the proposition follows.
\qed
\end{proof}

\begin{proposition}\label{prop:speedbound}
For the speed $s_{i,j}$ of any task $T_{i,j} \in \mathcal{T}$ in the optimal schedule it holds that
\begin{equation*}
\frac{v_{i,j}}{t_{\max}}\leq s_{i,j} \leq \left(\frac{E}{v_{i,j}}\right)^{\frac{1}{\beta-1}}
\end{equation*}
\end{proposition}
\begin{proof}
The processing time $p_{i,j}$ of a task $T_{i,j} \in
\mathcal{T}$ in an  optimal schedule cannot exceed the maximum
completion time, that is  $p_{i,j}=\frac{v_{i,j}}{s_{i,j}} \leq
C_{max}$ and,  since by Proposition~\ref{prop:T} it holds that
$C_{\max} \leq t_{\max}$,  the lower bound follows.

The energy consumption of any task cannot exceed the energy
budget, that is $E_{i,j}=v_{i,j} s_{i,j}^{\beta-1} \leq E$ and
the upper bound follows.
\qed
\end{proof}

Let $s_L=\frac{v_{\min}}{t_{\max}}$ and $s_U=\left(\frac{E}{v_{\min}}\right)^{1/(\beta-1)}$ be
an upper and a lower bound, respectively, on the speed of any task.
Given these bounds, we discretize the interval $[s_L,s_U]$ geometrically.
In other words, we assume that the processors can only run according to one of the following speeds:
$s_L,s_L(1+\epsilon),s_L(1+\epsilon)^2,\ldots,s_L(1+\epsilon)^k$,
where $k$ is the smallest integer such that $s_L(1+\epsilon)^k \geq s_U$.
Note that $k=\lceil\log_{1+\epsilon}\frac{s_U}{s_L}\rceil$ and hence the number of possible speeds is polynomial to the size of the instance and to $1/\epsilon$.
We denote by $\mathcal{V}=\{s_L(1+\epsilon)^{\ell}|\epsilon>0, 0 \leq \ell \leq k\}$ the set of all possible discrete speed values.
Let also $s_{\max}=s_L(1+\epsilon)^k$.

\begin{lemma}\label{le:discretespeeds}
There is a feasible $(1+\epsilon)$-approximate schedule for the \mr problem in which each task $T_{i,j} \in \mathcal{T}$ runs at a speed $s \in \mathcal{V}$.
\end{lemma}
\begin{proof}
Let an optimal schedule for our problem and consider the speed of each task $T_{i,j} \in \mathcal{T}$
rounded down to the closest $s_L(1+\epsilon)^{\ell}$ value.
As the speeds are decreased, the energy consumption of $\mathcal{S}$ does not exceed $E$.
Moreover, the execution time of all tasks, and hence the completion time of every job
and the optimal objective value increase by a factor at most $(1+\epsilon)$. \qed
\end{proof}

Henceforth we will consider the \mr problem in which each task $T_{i,j} \in \mathcal{T}$ runs at a single speed $s \in \mathcal{V}$.
We call this version of the problem \dmr.

\subsection{Linear Programming Relaxation}\label{se:lp}

In what follows we give an interval-indexed linear programming relaxation of the \dmr problem.
In order to do this, we discretize the time horizon of an optimal schedule as follows.
By Proposition~\ref{prop:T}, in any optimal schedule, all jobs are executed during the interval $(0,t_{\max}]$.
We partition $(0,t_{\max}]$ into the intervals
$(0,\lambda],(\lambda,\lambda(1+\delta)],(\lambda(1+\delta),\lambda(1+\delta)^2],\ldots,(\lambda(1+\delta)^{u-1},\lambda(1+\delta)^u]$,
where $\delta>0$ is a small constant, $\lambda>0$ is a constant that we will define later,
and $u$ is the smallest integer such that $\lambda(1+\delta)^{u-1} \geq t_{\max}$.
Let $\tau_0=0$ and $\tau_t=\lambda(1+\delta)^{t-1}$, for $1 \leq t \leq u+1$.
Moreover, let $I_t=(\tau_t,\tau_{t+1}]$, for $0 \leq t \leq u$, and
$|I_t|$ be the length of the interval $I_t$, i.e., $|I_0|=\lambda$ and $|I_t|=\lambda\delta(1+\delta)^{t-1}$, $1 \leq t \leq u$.
Note that, the number of intervals is polynomial to the size of the instance and to $1/\delta$, as $u=\lceil\log_{1+\delta}\frac{t_{\max}}{\lambda}\rceil+1$.

Let $p_{i,j,s} = \frac{v_{i,j}}{s}$ be the potential processing time for each task $T_{i,j} \in \mathcal{T}$
if it is executed entirely with speed $s \in \mathcal{V}$.
For each $T_{i,j}\in \mathcal{T}$, $t \in \{0,1,\ldots,u\}$ and $s\in \mathcal{V}$,
we introduce a variable $y_{i,j,s,t}$ that corresponds to the portion of the interval $I_t$ during which the task $T_{i,j}$ is executed with speed $s$.
In other words, $y_{i,j,s,t}|I_t|$ is the time that task $T_{i,j}$ is executed within the interval $I_t$ at speed $s$,
or equivalently, $\frac{y_{i,j,s,t}|I_t|}{p_{i,j,s}}$ is the fraction of the task $T_{i,j}$ that is executed within $I_t$ at speed $s$.
Note that the number of $y_{i,j,s,t}$ variables is polynomial to the size of the instance, to $1/\epsilon$ and to $1/\delta$.
Furthermore, for each task $T_{i,j} \in \mathcal{T}$, we introduce a variable $C_{i,j}$, which corresponds to the completion time of $T_{i,j}$.
Finally, let $C_j$, $j\in \mathcal{J}$, be the variable that corresponds to the completion time of job $j$.
(LP) is a linear programming relaxation of the \dmr problem.
\begin{figure}
\begin{alignat}{2}
 & (LP): \text{minimize} \sum_{j\in \mathcal{J}} w_j C_j \notag\\
 & \text{subject to}: \notag\\
 & \sum_{s\in \mathcal{V}} \sum_{t=0}^u \frac{y_{i,j,s,t}|I_t|}{p_{i,j,s}} = 1, ~~~~~~~~~~~~~~~~~~~~~~~ \forall T_{i,j}\in \mathcal{T} \label{lp:p1}\\
 & \sum_{j:T_{i,j}\in \mathcal{T}}\sum_{s\in \mathcal{V}} y_{i,j,s,t} \leq 1, ~~~~~~~~~~~~~ \forall i\in \mathcal{P}, 0 \leq t \leq u \label{lp:p2}\\
 & C_{i,j} \geq \frac{1}{2} \sum_{s\in \mathcal{V}} y_{i,j,s,0}|I_0| \left(\frac{1}{p_{i,j,s}} + 1 \right) + \notag\\
 & ~~ \sum_{t=1}^u \sum_{s\in \mathcal{V}} \left(\frac{y_{i,j,s,t}|I_t|}{p_{i,j,s}}\tau_t + \frac{1}{2} y_{i,j,s,t}|I_t|\right), ~~~ \forall T_{i,j}\in \mathcal{T} \label{lp:p3}\\
 & C_{j} \geq C_{i,j}, ~~~~~~~~~~~~~~~~~~~~~~~~~~~~~~~~~~~~~~ \forall T_{i,j}\in \mathcal{T} \label{lp:p4}\\
 & \sum_{T_{i,j}\in\mathcal{T}}\sum_{s\in\mathcal{V}}\sum_{t=0}^u y_{i,j,s,t}|I_t| s^{\beta} \leq E \label{lp:p5}\\
 & \sum_{t=0}^{\ell} \sum_{s \in \mathcal{V}} \frac{y_{i,j,s,t}|I_t|}{p_{i,j,s}} \geq \sum_{t=0}^{\ell} \sum_{s \in \mathcal{V}} \frac{y_{i',j,s,t}|I_t|}{p_{i',j,s}}, \notag\\
 & ~~~~~~~~~~~~~~~~~~~~~~~~~ \forall T_{i,j} \in \mathcal{M}, T_{i',j} \in \mathcal{R}, 0 \leq \ell \leq u\label{lp:p6}\\
 & y_{i,j,s,t} = 0, ~~~~~~~~~~~~~~~~ \forall T_{i,j} \in \mathcal{T}, s \in \mathcal{V}, t:\tau_t<r_j \label{lp:p7}\\
 & y_{i,j,s,t}, C_{i,j}, C_j \geq 0, ~~~~~~ \forall T_{i,j} \in \mathcal{T}, s \in \mathcal{V}, 0 \leq t \leq u
\end{alignat}
\end{figure}

Our objective is to minimize the sum of weighted completion times of all jobs.
For each task $T_{i,j} \in \mathcal{T}$, the corresponding constraint~(\ref{lp:p1}) ensures that $T_{i,j}$ is entirely executed.
Constraints~(\ref{lp:p2}) enforce that the total amount of processing time that is executed within an interval $I_t$ cannot exceed its length.
In~\cite{SchulzS02}, the authors proposed a lower bound for the completion time of a job.
This lower bound can be adapted to our problem and for the completion time
of a task $T_{i,j} \in \mathcal{T}$ leads to a corresponding constraint~(\ref{lp:p3}).
Constraints~(\ref{lp:p4}) ensure that the completion time of each job is the maximum over the completion times of all its tasks.
Constraint~(\ref{lp:p5}) ensures that the given energy budget is not exceeded.
Note that the value $s^{\beta}$ for each $s \in \mathcal{V}$ is a fixed number.
Constraints~(\ref{lp:p6}) imply the precedence constraints between the Map and the Reduce tasks of the same job,
as they enforce that the fraction of a Map task that is executed up to each time point
should be at least the fraction of a Reduce task of the same job executed up to the same time point;
hence, each Map task completes before all Reduce tasks of the same job.
Constraints~(\ref{lp:p7}) do not allow tasks of a job to be executed before their release date.


In what follows, we denote an optimal solution to (LP) by $(\bar{y}_{i,j,s,t}, \bar{C}_{i,j}, \bar{C}_j)$.

\subsection{The Algorithm}

In this section we use (LP) to derive a feasible schedule for the \dmr problem.
Depending on the choice of some parameters, this schedule may exceed the energy budget.
As we show, there is a tradeoff where a constant factor approximation ratio can be derived.

Our algorithm is based on the idea of list scheduling in order of $\alpha$-points~\cite{HallSSW97}.
In general, an $\alpha$-point of a job is the first point in time
where an $\alpha$-fraction of the job has been completed, where $\alpha \in (0,1)$ is a constant that depends on the analysis.
In this paper, we will define the $\alpha$-point $t_{i,j}^{\alpha}$ of a task $T_{i,j} \in \mathcal{T}$ as
the minimum $\ell$, $0 \leq \ell \leq u$, such that at least an $\alpha$-fraction of $v_{i,j}$ is accomplished up to the interval $I_{\ell}$ to (LP), i.e.,
\begin{equation*}
t_{i,j}^{\alpha} = \min\Bigg\{\ell: \sum_{t=0}^{\ell} \sum_{s \in \mathcal{S}} \frac{\bar{y}_{i,j,s,t} |I_t|}{p_{i,j,s}} \geq \alpha \Bigg\}\label{eq:a-point}
\end{equation*}
Thus, once our algorithm has computed an optimal solution $(\bar{y}_{i,j,s,t}$, $\bar{C}_{i,j}$, $\bar{C}_j)$ to (LP),
it calculates the corresponding $\alpha$-point, $t_{i,j}^{\alpha}$, for each task $T_{i,j}\in \mathcal{T}$.
Then, combining the ideas of~\cite{ChenKL12} with the notion of $\alpha$-points~\cite{HallSSW97}, we create a feasible schedule as follows:
For each processor $i \in \mathcal{P}$, we consider a priority list $\sigma_i$ of its tasks such that tasks with smaller $\alpha$-point have higher priority.
A crucial point in our analysis is that we consider that a task $T_{i,j} \in \mathcal{T}$
becomes \emph{available} for the algorithm after the time $\tau_{t_{i,j}^{\alpha}+1} > r_j$.
Moreover, if $T_{i,j} \in \mathcal{R}$ then we need also all tasks $T_{i',j} \in \mathcal{M}$ to be completed in order $T_{i,j}$ to be considered as available.
For each task $T_{i,j} \in \mathcal{T}$, we use a constant speed $s_{i,j}=\frac{v_{i,j}}{p_{i,j}}$, where
\begin{equation*}
p_{i,j} = \gamma \sum_{t=0}^{t_{i,j}^{\alpha}} \sum_{s\in\mathcal{V}} \bar{y}_{i,j,s,t}|I_t|
\end{equation*}
is the processing time of $T_{i,j}$ used by our algorithm,
and $\gamma>0$ is a constant that we define later and describes the tradeoff between the energy consumption and the weighted completion time of jobs.
At each time point where a processor $i \in \mathcal{P}$ is available,
our algorithm selects the highest priority available task in $\sigma_i$ which has not been yet executed.
Note that our algorithm always create a feasible solution as we do not insist on selecting the highest priority task if this is not available.
\algomr gives a formal description of our algorithm.

\begin{algorithm}
\algomr
\begin{algorithmic}[1]
\STATE Compute an optimal solution $(\bar{y}_{i,j,s,t}, \bar{C}_{i,j}, \bar{C}_j)$ to $(LP)$.
\FOR {each task $T_{i,j}\in\mathcal{T}$}
\STATE Compute the $\alpha$-point $\displaystyle t_{i,j}^{\alpha}$, the processing time $p_{i,j}$ and the speed $s_{i,j}$.
\ENDFOR
\FOR {each processor $i \in \mathcal{P}$}
\STATE Compute the priority list $\sigma_i$.
\ENDFOR
\FOR {each time where a processor $i\in \mathcal{P}$ becomes available}
\STATE Select the first available task, let $T_{i,j}$, in $\sigma_i$ which has not been yet executed.
\STATE Schedule $T_{i,j}$, non-preemptively, with processing time $p_{i,j}$.\\Let $C_{i,j}$ be the completion time of task $T_{i,j}$.
\ENDFOR
\FOR {each job $j\in \mathcal{J}$}
\STATE Compute its completion time $C_j = \max_{i\in\mathcal{P}}C_{i,j}$.
\ENDFOR
\end{algorithmic}
\end{algorithm}

Note that the processing time of a task $T_{i,j} \in \mathcal{T}$ to an optimal solution to (LP) is
\begin{equation*}
\bar{p}_{i,j}=\sum_{t=0}^u \sum_{s\in\mathcal{V}} \bar{y}_{i,j,s,t}|I_t|
\end{equation*}
Hence, the energy consumption $\bar{E}_{i,j}=\sum_{s\in\mathcal{V}}\sum_{t=0}^u \bar{y}_{i,j,s,t}|I_t| s^{\beta}$
for the execution of $T_{i,j}$ to an optimal solution to (LP) may be smaller or bigger than
the energy consumption $E_{i,j}$ for the execution of $T_{i,j}$ by the algorithm.
In order to give the relation between these two quantities, we need the following technical lemma.

\begin{lemma} \label{le:technical}
Let $s_1,s_2,\ldots,s_k$ and $a_1,a_2,\ldots,a_k$ be positive values and $\beta>2$.
Then, it holds that
\begin{equation*}
\left(\frac{1}{\sum_{i=1}^ka_i\frac{1}{s_i}}\right)^{\beta-1}\leq \frac{\sum_{i=1}^ka_is_i^{\beta-1}}{\left(\sum_{i=1}^ka_i\right)^{\beta}}
\end{equation*}
\end{lemma}
\begin{proof}
The expression of the statement can be written equivalently as follows.
\begin{equation}
\label{Eq:Technical}
\left(\frac{\sum_{i=1}^ka_i}{\sum_{i=1}^ka_i\frac{1}{s_i}}\right)^{\beta-1}\leq \frac{\sum_{i=1}^ka_is_i^{\beta-1}}{\sum_{i=1}^ka_i}
\end{equation}
Note that the function $f(x)=x^{\beta-1}$ is convex for $\beta>2$.
Thus, by the Jensen's inequality we have that
\begin{equation*}
f\left(\frac{\sum_{i=1}^ka_is_i}{\sum_{i=1}^ka_i}\right)\leq \frac{\sum_{i=1}^ka_if(s_i)}{\sum_{i=1}^ka_i}
\end{equation*}
which is translated as
\begin{equation*}
\left(\frac{\sum_{i=1}^ka_is_i}{\sum_{i=1}^ka_i}\right)^{\beta-1}\leq \frac{\sum_{i=1}^ka_is_i^{\beta-1}}{\sum_{i=1}^ka_i}
\end{equation*}
Therefore, in order to show inequality (\ref{Eq:Technical}), it suffices to show that
\begin{equation*}
\left(\frac{\sum_{i=1}^ka_i}{\sum_{i=1}^ka_i\frac{1}{s_i}}\right)^{\beta-1}\leq \left(\frac{\sum_{i=1}^ka_is_i}{\sum_{i=1}^ka_i}\right)^{\beta-1}
\end{equation*}
Thus, it suffices to prove that
\begin{equation*}
\frac{\sum_{i=1}^ka_i}{\sum_{i=1}^ka_i\frac{1}{s_i}}\leq \frac{\sum_{i=1}^ka_is_i}{\sum_{i=1}^ka_i}
\end{equation*}
An equivalent representation of the above expression is
\begin{eqnarray*}
\left(\sum_{i=1}^ka_i\right)^2\leq\left(\sum_{i=1}^ka_is_i\right) \left(\sum_{i=1}^ka_i\frac{1}{s_i}\right)\Leftrightarrow \\
\sum_{i=1}^ka_i^2+\sum_{i,j=1,\;i\neq j}^k2a_ia_j\leq \sum_{i=1}^ka_i^2 \sum_{i,j=1,\;i\neq j}^ka_ia_j\left(\frac{s_i}{s_j}+\frac{s_j}{s_i}\right)
\end{eqnarray*}
The last inequality is always true, as
\begin{equation*}
2\leq\frac{s_i}{s_j}+\frac{s_j}{s_i}\Leftrightarrow2\leq\frac{s_i^2+s_j^2}{s_is_j} \Leftrightarrow 0\leq(s_i-s_j)^2
\end{equation*}
and hence the lemma follows.\qed
\end{proof}

\begin{lemma}\label{le:speed}
Let $\bar{E}_{i,j}$ and $E_{i,j}$ be the energy consumption of the task $T_{i,j}\in \mathcal{T}$
to the optimal solution to (LP) and to the solution of \algomr, respectively.
It holds that
\begin{equation*}
E_{i,j}\leq \frac{1}{\gamma^{\beta-1}\alpha^{\beta}} \bar{E}_{i,j}
\end{equation*}
\end{lemma}
\begin{proof}
By the definition of $E_{i,j}$ we have that
\begin{eqnarray*}
E_{i,j} & = & v_{i,j}s_{i,j}^{\beta-1}
 = v_{i,j} \left(\frac{v_{i,j}}{p_{i,j}}\right)^{\beta-1}\\
 & = & v_{i,j}\left(\frac{v_{i,j}}{\gamma \sum_{s\in\mathcal{V}}\sum_{t=0}^{t_{i,j}^{\alpha}}\bar{y}_{i,j,s,t}|I_t|}\right)^{\beta-1}
\end{eqnarray*}
Since for each speed $s\in\mathcal{V}$, $p_{i,j,s} = \frac{v_{i,j}}{s}$, the above equality can be written as
\begin{equation*}
E_{i,j} = \frac{v_{i,j}}{\gamma^{\beta-1}}\left(\frac{1}{\sum_{s\in\mathcal{V}}\frac{1}{s}\sum_{t=0}^{t_{i,j}^{\alpha}}\frac{\bar{y}_{i,j,s,t}|I_t|}{p_{i,j,s}}}\right)^{\beta-1}
\end{equation*}
Hence, by using Lemma~\ref{le:technical} we get
\begin{equation*}
E_{i,j} \leq \frac{v_{i,j}}{\gamma^{\beta-1}} \cdot \frac{\sum_{s\in\mathcal{V}}s^{\beta-1}\sum_{t=0}^{t_{i,j}^{\alpha}}\frac{\bar{y}_{i,j,s,t}|I_t|}{p_{i,j,s}}}{\left(\sum_{s\in\mathcal{V}}\sum_{t=0}^{t_{i,j}^{\alpha}}\frac{\bar{y}_{i,j,s,t}|I_t|}{p_{i,j,s}}\right)^{\beta}}
\end{equation*}
By the definition of $\alpha$-points we have that
$\sum_{t=0}^{t_{i,j}^{\alpha}} \sum_{s\in \mathcal{V}}\frac{\bar{y}_{i,j,s,t}|I_t|}{p_{i,j,s}} \geq \alpha$, and thus
\begin{eqnarray*}
E_{i,j} & \leq & \frac{1}{\gamma^{\beta-1}\alpha^{\beta}} \sum_{s\in\mathcal{V}}s^{\beta-1}\sum_{t=0}^{t_{i,j}^{\alpha}}v_{i,j}\frac{\bar{y}_{i,j,s,t}|I_t|}{p_{i,j,s}}\\
 & = & \frac{1}{\gamma^{\beta-1}\alpha^{\beta}} \sum_{s\in\mathcal{V}}s^{\beta-1}\sum_{t=0}^{t_{i,j}^{\alpha}}v_{i,j}\frac{\bar{y}_{i,j,s,t}|I_t|}{v_{i,j}/s} \\
 & = & \frac{1}{\gamma^{\beta-1}\alpha^{\beta}} \sum_{s\in\mathcal{V}}\sum_{t=0}^{t_{i,j}^{\alpha}}\bar{y}_{i,j,s,t}|I_t|s^{\beta}\\
 & \leq & \frac{1}{\gamma^{\beta-1}\alpha^{\beta}} \sum_{s\in\mathcal{V}}\sum_{t=0}^u\bar{y}_{i,j,s,t}|I_t|s^{\beta}
 = \frac{1}{\gamma^{\beta-1}\alpha^{\beta}} \bar{E}_{i,j}
\end{eqnarray*}
and the lemma follows.\qed
\end{proof}

The following lemma provides a lower bound to the completion time $\bar{C}_{i,j}$ of the task $T_{i,j} \in \mathcal{T}$ given by the (LP).
\begin{lemma}
If $\lambda < \alpha \frac{v_{\min}}{s_{\max}}$, then
for each task $T_{i,j} \in \mathcal{T}$ it holds that $\bar{C}_{i,j}\geq(1-\alpha)\cdot\tau_{t_{i,j}^{\alpha}}$.
\label{le:Clowerbound}
\end{lemma}
\begin{proof}
Recall that $t_{i,j}^{\alpha}$ corresponds to the interval $I_{t_{i,j}^{\alpha}} = (\tau_{t_{i,j}^{\alpha}},\tau_{t_{i,j}^{\alpha}+1}]$.
If we select $\lambda < \alpha \frac{v_{\min}}{s_{\max}}$, then there is no task with $\alpha$-point to the interval $I_0$.
Hence, we can consider that the $\alpha$-point of each task $T_{i,j} \in \mathcal{T}$
corresponds to an interval of the form $(\lambda(1+\delta)^{t_{i,j}^{\alpha}-1},\lambda(1+\delta)^{t_{i,j}^{\alpha}}]$.

Starting from constraint~(\ref{lp:p3}) we have that
\begin{eqnarray*}
\bar{C}_{i,j} & \geq & \frac{1}{2} \sum_{s\in \mathcal{V}} \bar{y}_{i,j,s,0}|I_0| \left(\frac{1}{p_{i,j,s}} + 1 \right) \\
 && +\sum_{t=1}^u \sum_{s\in \mathcal{V}} \left(\frac{\bar{y}_{i,j,s,t}|I_t|}{p_{i,j,s}}\tau_t + \frac{1}{2} \bar{y}_{i,j,s,t}|I_t|\right)\\
 & \geq & \sum_{t=t_{i,j}^{\alpha}}^u \sum_{s\in \mathcal{V}} \left(\frac{\bar{y}_{i,j,s,t}|I_t|}{p_{i,j,s}}\tau_t + \frac{1}{2} \bar{y}_{i,j,s,t}|I_t|\right)\\
 & \geq & \sum_{t=t_{i,j}^{\alpha}}^u \sum_{s\in \mathcal{V}} \frac{\bar{y}_{i,j,s,t}|I_t|}{p_{i,j,s}}\tau_t \\
 & \geq & \tau_{t_{i,j}^{\alpha}} \sum_{t=t_{i,j}^{\alpha}}^u \sum_{s\in \mathcal{V}} \frac{\bar{y}_{i,j,s,t}|I_t|}{p_{i,j,s}}
 \geq (1-\alpha) \cdot \tau_{t_{i,j}^{\alpha}}
\end{eqnarray*}
where the last inequality holds by constraint~(\ref{lp:p1}) and as by the definition of $\alpha$-point we know that
$\sum_{t=0}^{t_{i,j}^{\alpha}-1} \sum_{s\in \mathcal{V}} \frac{\bar{y}_{i,j,s,t}|I_t|}{p_{i,j,s}} < \alpha$.
\qed
\end{proof}

The following theorem gives the approximation ratio of \algomr.

\begin{theorem}
\algomr is a $\frac{1}{\gamma^{\beta-1}\alpha^{\beta}}$-energy $\frac{\gamma^2+3\gamma+1}{1-\alpha}(1+\delta)$-approximation algorithm for the \dmr problem,
where $\alpha\in(0,1)$, $\gamma>0$ and $\delta>0$.
\label{thm:main}
\end{theorem}
\begin{proof}
Consider the schedule $\mathcal{S}$ produced by \algomr and let $T_{i,j} \in \mathcal{M}$ be any Map task.
Recall that $\sigma_i$ is the priority list of processor $i$.
Let $\sigma_i(j) \subseteq \sigma_i$ be the list of tasks with priority higher than the priority of $T_{i,j}$ in $\sigma_i$, including $T_{i,j}$.
Then, for $C_{i,j}$ it holds that
\begin{equation}
C_{i,j} \leq \tau_{t_{i,j}^{\alpha}+1} + \sum_{k\in\sigma_i(j)}p_{i,k} \label{eq:up}
\end{equation}
as $T_{i,j}$ is always available after $\tau_{t_{i,j}^{\alpha}+1}$, as a Map task.
For the total processing time of jobs in $\sigma_i(j)$ we have that
\begin{eqnarray*}
\sum_{k\in\sigma_i(j)}p_{i,k}
 & = & \sum_{k\in\sigma_i(j)} \gamma \sum_{t=0}^{t_{i,k}^{\alpha}} \sum_{s\in\mathcal{V}} \bar{y}_{i,k,s,t}|I_t|\\
 & \leq & \gamma \sum_{k\in\sigma_i(j)} \sum_{t=0}^{t_{i,j}^{\alpha}} \sum_{s\in\mathcal{V}} \bar{y}_{i,k,s,t}|I_t| \\
 & \leq & \gamma \sum_{k\in\sigma_i} \sum_{t=0}^{t_{i,j}^{\alpha}} \sum_{s\in\mathcal{V}} \bar{y}_{i,k,s,t}|I_t|\\
 & = & \gamma \sum_{t=0}^{t_{i,j}^{\alpha}}|I_t| \sum_{k\in\sigma_i} \sum_{s\in\mathcal{V}} \bar{y}_{i,k,s,t}
 \leq \gamma \sum_{t=0}^{t_{i,j}^{\alpha}} |I_t|
 = \gamma \tau_{t_{i,j}^{\alpha}+1}
\end{eqnarray*}
where the last inequality holds by applying constraint~(\ref{lp:p2}) of the (LP).
Thus, from inequality~(\ref{eq:up}) we have
\begin{equation}
C_{i,j} \leq (\gamma+1)\tau_{t_{i,j}^{\alpha}+1} \label{eq:map}
\end{equation}
for each Map task $T_{i,j} \in \mathcal{T}$.

Consider now a job $j \in \mathcal{J}$ and let $T_{i,j} \in \mathcal{R}$ be a Reduce task of $j$.
Moreover, let $T_{i',j} \in \mathcal{M}$ be the Map task of $j$ that completes last in $\mathcal{S}$,
i.e., $C_{i',j}=\max\{C_{i,j}:T_{i,j}\in \mathcal{M}, i \in \mathcal{P}\}$.
By definition, $T_{i,j}$ becomes available at time $t=\max\{\tau_{t_{i,j}^{\alpha}+1},C_{i',j}\}$.
Note that
\begin{equation*}
t \leq \max\{\tau_{t_{i,j}^{\alpha}+1},(\gamma+1)\tau_{t_{i',j}^{\alpha}+1}\}
\leq \max\{\tau_{t_{i,j}^{\alpha}+1},(\gamma+1)\tau_{t_{i,j}^{\alpha}+1}\}
= (\gamma+1)\tau_{t_{i,j}^{\alpha}+1}
\end{equation*}
where the first inequality holds by inequality~(\ref{eq:map}) and the second by the constraint~(\ref{lp:p6}) of (LP).

Let again $\sigma_i(j)$ be the list of tasks with higher priority than $T_{i,j}$ in $\sigma_i$, including $T_{i,j}$.
If in the schedule $\mathcal{S}$ the processor $i$ at time $t$ executes a task $T_{i,j'} \not\in \sigma_i(j)$,
then for the completion time of $T_{i,j}$ it holds that
\begin{equation}
C_{i,j} \leq t + p_{i,j'} + \sum_{k\in\sigma_i(j)} p_{i,k} \label{eq:red}
\end{equation}
because $T_{i,j}$ is available after time $t$ and it has higher priority than any task $T_{i,j''} \not \in \sigma_i(j)$.
As before, we have that
\begin{equation*}
\sum_{k\in\sigma_i(j)} p_{i,k} \leq \gamma \tau_{t_{i,j}^{\alpha}+1}
\end{equation*}
Moreover, for the processing time of $T_{i,j'}$ it holds that
\begin{equation*}
p_{i,j'} = \gamma \sum_{t=0}^{t_{i,j'}^{\alpha}} \sum_{s \in \mathcal{V}} \bar{y}_{i,j',s,t} |I_t| \leq \gamma \tau_{t_{i,j'}^{\alpha}+1} < \gamma t
\end{equation*}
as $T_{i,j'}$ is executed at time $t$ and hence it is available.
Then, by equation~(\ref{eq:red}) we have
\begin{equation*}
C_{i,j} \leq t + \gamma t + \gamma\tau_{t_{i,j}^{\alpha}+1} \leq ((\gamma+1)^2+\gamma)\tau_{t_{i,j}^{\alpha}+1} = (\gamma^2+3\gamma+1) \tau_{t_{i,j}^{\alpha}+1}
\end{equation*}

As $\tau_{t_{i,j}^{\alpha}+1} = (1+\delta)\tau_{t_{i,j}^{\alpha}}$, using Lemma~\ref{le:Clowerbound} we get
\begin{equation*}
C_{i,j} \leq \frac{\gamma^2+3\gamma+1}{1-\alpha}(1+\delta)\bar{C}_{i,j},\label{eq:cij}
\end{equation*}
and by using constraint (\ref{lp:p4}) of (LP)
\begin{equation*}
C_{i,j} \leq \frac{\gamma^2+3\gamma+1}{1-\alpha}(1+\delta)\bar{C}_{j}
\end{equation*}
Since the above inequality holds for each processor $i \in \mathcal{P}$,
it must also hold for $C_j = \max_{i\in\mathcal{P}}\{C_{i,j}\}$ and thus
\begin{equation*}
C_j \leq \frac{\gamma^2+3\gamma+1}{1-\alpha}(1+\delta)\bar{C}_j
\end{equation*}
If we sum up all weighted completion times in $\mathcal{S}$ we yield
\begin{equation*}
\sum_{j\in \mathcal{J}}w_jC_j \leq \frac{\gamma^2+3\gamma+1}{1-\alpha}(1+\delta)\sum_{j\in \mathcal{J}}w_j\bar{C}_j
\end{equation*}
and as $\sum_{j\in \mathcal{J}}w_j\bar{C}_j$ is a lower bound to the objective value of an optimal solution
for the \dmr problem, the theorem follows. \qed
\end{proof}

Note that in the absence of precedence constraints between the tasks, the above analysis can be improved.
Indeed, we can consider that all tasks are Map tasks,
and hence an upper bound to their completion time into the schedule created by \algomr is given by Inequality~(\ref{eq:map}).
Then, the following corollary holds.

\begin{corollary} \label{cor:noprec}
\algomr is a $\frac{1}{\gamma^{\beta-1}\alpha^{\beta}}$-energy $\frac{\gamma+1}{1-\alpha}(1+\delta)$-approximation algorithm for the \dmr problem without precedence constraints,
where $\alpha\in(0,1)$, $\gamma>0$ and $\delta>0$.
\end{corollary}

Moreover, in the absence of both precedence constraints between tasks and release dates of jobs, our analysis can be further improved.
As before, we can consider that all tasks are Map tasks.
In addition, we can drop the demand that a task $T_{i,j} \in \mathcal{T}$ becomes available for the algorithm after the time $\tau_{t_{i,j}^{\alpha}+1}$.
Hence, Inequality~(\ref{eq:up}) is simplified to $C_{i,j} \leq \sum_{k \in \sigma_i(j)} p_{i,k}$,
as all tasks are released at time 0 and they are available at any time.
Then, the following corollary holds.

\begin{corollary} \label{cor:noprecreal}
\algomr is a $\frac{1}{\gamma^{\beta-1}\alpha^{\beta}}$-energy $\frac{\gamma}{1-\alpha}(1+\delta)$-approximation algorithm
for the \dmr problem without precedence constraints and release dates,
where $\alpha\in(0,1)$, $\gamma>0$ and $\delta>0$.
\end{corollary}

By combining Lemma~\ref{le:discretespeeds}, Theorem~\ref{thm:main} and Corollaries~\ref{cor:noprec} and~\ref{cor:noprecreal},
and as we can select an $\varepsilon$ such that $(1+\delta)(1+\epsilon)\leq(1+\varepsilon)$, the following theorem holds.

\begin{theorem}
There is a $\frac{1}{\gamma^{\beta-1}\alpha^{\beta}}$-energy $\frac{\gamma^2+3\gamma+1}{1-\alpha}(1+\varepsilon)$-approximation algorithm for the \mr problem,
a $\frac{1}{\gamma^{\beta-1}\alpha^{\beta}}$-energy $\frac{\gamma+1}{1-\alpha}(1+\varepsilon)$-approximation algorithm for the \mr problem without precedence constraints,
and a $\frac{1}{\gamma^{\beta-1}\alpha^{\beta}}$-energy $\frac{\gamma}{1-\alpha}(1+\varepsilon)$-approximation algorithm for the \mr problem without precedence constraints and release dates,
where $\alpha\in(0,1)$, $\gamma>0$ and $\varepsilon>0$.
\end{theorem}

In Fig.\ref{fig:tradeoff} we depict a tradeoff between energy augmentation and approximation ratio for some practical values of $\beta$.
Note that, by choosing $\gamma=\frac{1}{\alpha \sqrt[\beta-1]{\alpha}}$, energy augmentation is not allowed and the schedule can be converted to a constant-factor approximate schedule.
In this case the following theorem holds.

\begin{theorem}
There is a $\frac{(\alpha \sqrt[\beta-1]{\alpha})^2+3\alpha \sqrt[\beta-1]{\alpha}+1}{(\alpha \sqrt[\beta-1]{\alpha})^2(1-\alpha)}(1+\varepsilon)$-approximation algorithm for the \mr problem,
a $\frac{\alpha \sqrt[\beta-1]{\alpha}+1}{\alpha \sqrt[\beta-1]{\alpha}(1-\alpha)}(1+\varepsilon)$-approximation algorithm for the \mr problem without precedence constraints,
and a $\frac{1}{\alpha \sqrt[\beta-1]{\alpha}(1-\alpha)}(1+\varepsilon)$-approximation algorithm for the \mr problem without precedence constraints and release dates,
where $\alpha\in(0,1)$ and $\varepsilon>0$.
\end{theorem}

The ratios of the above theorem can be optimized by selecting the appropriate value of $\alpha$ for each $\beta$.
Table~\ref{tbl:results} gives the achieved ratios for practical values of $\beta$.

\begin{table}[htb]
\begin{center}
\begin{tabular}{c||c|c|c}
  $\beta$ & general & without precedence & without precedence \& without release dates \\
  \hline
  \hline
  2   & 37.52 & 9.44 & 6.75 \\
  2.2 & 34.89 & 8.84 & 6.29 \\
  2.4 & 33.01 & 8.41 & 5.97 \\
  2.6 & 31.59 & 8.09 & 5.72 \\
  2.8 & 30.50 & 7.84 & 5.53 \\
  3   & 29.62 & 7.64 & 5.38
\end{tabular}
\end{center}
\caption{Approximation ratios for the \mr problem for different values of $\beta$.}
\label{tbl:results}
\end{table}

\section{A Convex Programming Approach}\label{se:cp}

We are interested in natural list scheduling policies such as \textsc{First Come First Serve} (\fcfs) and \textsc{Smith Rule} (\sr). However, in our context we need to determine the speeds of every task in order to respect the energy budget. For that, we propose a convex programming relaxation of our problem when an order of the jobs is prespecified.

\subsection{The Convex Program}

Let $\sigma=\langle1,2,\ldots,n\rangle$ be a given order of the jobs.
Consider now the restricted version of the \mr problem where
for each processor $i \in \mathcal{P}$ the tasks are forced to be executed according to this order.
We shall refer to this problem as the \mrs problem.
Note that, the order is the same for all processors.
We write $j \prec j'$ if job $j \in \mathcal{J}$ precedes job $j' \in \mathcal{J}$ in $\sigma$.
We propose a convex program that considers the order $\sigma$ as input and returns a solution that is
a lower bound to the optimal solution for the \mrs problem.

In order to formulate our problem as a convex program, let $p_{i,j}$ be a variable that corresponds to the processing time of task $T_{i,j} \in \mathcal{T}$.
Moreover, for each task $T_{i,j} \in \mathcal{T}$, we introduce a variable $C_{i,j}$ that determines the completion time of $T_{i,j}$.
Finally, let $C_j$, $j\in \mathcal{J}$, be the variable that corresponds to the completion time of job $j$.
Consider the following convex programming formulation of the \mrs problem.
\begin{align}
 & (CP): \text{minimize} \sum_{j \in \mathcal{J}} w_j C_j && \notag\\
 & \text{subject to}: && \notag\\
 & \sum_{T_{i,j} \in \mathcal{T}} \frac{v_{i,j}^{\beta}}{p_{i,j}^{\beta-1}}\leq E && \label{cp:p1}\\
 & r_{j'} + \sum_{k=j'}^j p_{i,k} \leq C_{i,j}, && \forall T_{i,j}, T_{i,j'} \in \mathcal{T}, j' \prec j \label{cp:p2}\\
 & C_{i',j} + p_{i,j} \leq C_{i,j}, && \forall T_{i,j} \in \mathcal{R}, T_{i',j} \in \mathcal{M} \label{cp:p3}\\
 & C_{i,j} \leq C_j, && \forall T_{i,j} \in \mathcal{T} \label{cp:p4}\\
 & s_{i,j}, C_{i,j}, C_j \geq 0, && \forall T_{i,j} \in \mathcal{T}, j \in \mathcal{J} \notag
\end{align}

The objective function of (CP) is to minimize the weighted completion time of all jobs.
Constraint~(\ref{cp:p1}) guarantees that the energy budget is not exceeded.
Constraints~(\ref{cp:p2}) and~(\ref{cp:p3}) give lower bounds on the completion time of each task $T_{i,j} \in \mathcal{T}$,
based on the release dates and the precedence constraints, respectively.
Note that, if we do not consider precedences between the tasks,
then (CP) will return the optimal value of the objective function, instead of a lower bound of it,
as constraints~(\ref{cp:p2}) describe in a complete way the completion times of the tasks.
However, this is not true for constraints~(\ref{cp:p3}) which are responsible for the precedence constraints.
Finally, constraints~(\ref{cp:p4}) ensure that the completion time of each job is the maximum over the completion times among all of its tasks.

As the optimal solution to (CP) does not necessarily describe a feasible schedule,
we need to apply an algorithm that uses the processing times found by (CP) and the order $\sigma$ so as to create a feasible schedule for the \mrs problem,
and hence for the \mr problem.
In fact, it suffices to apply, for example, the Lines~6-8 of \algomr, by considering the same order for all processors.

\begin{figure}[ht]
\centering
\begin{minipage}[b]{0.45\linewidth}
\begin{tikzpicture}[scale=0.6]
\begin{axis}[xlabel=approximation ratio,ylabel=energy augmentation (\%),ymin=0]
\addplot[smooth,thick] plot coordinates {
 (37.52,0)
 (33.32,10)
 (29.97,20)
 (27.25,30)
 (24.99,40)
 (23.10,50)
 (21.49,60)
 (20.10,70)
 (18.90,80)
 (17.84,90)
 (16.91,100)
};
\addplot[smooth,dashed,thick] plot coordinates {
 (32.25,0)
 (29.80,10)
 (27.75,20)
 (26.02,30)
 (24.53,40)
 (23.24,50)
 (22.11,60)
 (21.11,70)
 (20.22,80)
 (19.42,90)
 (18.69,100)
};
\addplot[smooth,dotted,thick] plot coordinates {
 (29.62,0)
 (27.91,10)
 (26.46,20)
 (25.20,30)
 (24.10,40)
 (23.13,50)
 (22.27,60)
 (21.49,70)
 (20.79,80)
 (20.15,90)
 (19.57,100)
};
\legend{$\beta=2~~$\\$\beta=2.5$\\$\beta=3~~$\\}
\end{axis}
\end{tikzpicture}
\caption{Tradeoff between energy augmentation and approximation ratio when $\beta=\{2,2.5,3\}$.}
\label{fig:tradeoff}
\end{minipage}
\quad
\begin{minipage}[b]{0.45\linewidth}
\begin{tikzpicture}[scale=0.6]
\begin{axis}[xlabel=number of jobs,ylabel=$\sum w_j C_j$,ymin=0,legend pos=north west]]
\addplot[smooth,mark=x,thick] plot coordinates {
 (5,1919/1000)
 (10,12857/1000)
 (15,38352/1000)
 (20,85177/1000)
 (25,174398/1000)
};
\addplot[smooth,mark=x,dashed,thick] plot coordinates {
 (5,2321/1000)
 (10,16222/1000)
 (15,47570/1000)
 (20,100876/1000)
 (25,200904/1000)
};
\addplot[smooth,thick] plot coordinates {
 (5,1259/1000)
 (10,8594/1000)
 (15,26120/1000)
 (20,60228/1000)
 (25,119361/1000)
};
\addplot[smooth,dashed,thick] plot coordinates {
 (5,1001/1000)
 (10,6311/1000)
 (15,18943/1000)
 (20,41997/1000)
 (25,86365/1000)
};
\legend{\fcfs\\\sr\\CP(\fcfs)\\CP(\sr)\\}
\end{axis}
\end{tikzpicture}
\caption{Experimental comparison of the solutions of \fcfs and \sr (scaled down by a factor of $10^3$).}
\label{fig:exp}
\end{minipage}
\end{figure}

\subsection{Scheduling Policies}

In this section we propose different orders of jobs and we discuss how far is an optimal solution for the \mrs problem using these orders
with respect to the optimal solution for the \mr problem.
Two standard orders of jobs are the following.
\bigskip

\noindent\textsc{First Come First Serve} (\fcfs): for each pair of jobs $j,j' \in \mathcal{J}$, if $r_j<r_{j'}$ then $j \prec j'$ in $\sigma$.
\smallskip

\noindent\textsc{Smith Rule} (\sr): for each pair of jobs $j,j' \in \mathcal{J}$,
if $\frac{w_j}{\sum_{T_{i,j} \in j} v_{i,j}} > \frac{w_{j'}}{\sum_{T_{i,j'} \in j'} v_{i,j'}}$ then $j \prec j'$ in $\sigma$.
\bigskip

%
%

The following propositions present negative results concerning the approximation ratio that we can achieve if we use the \fcfs or the \sr order.

\begin{proposition}\label{prop:fcfs}
Let $OPT$ and $OPT_{\fcfs}$ be the optimal solutions for the \mr and the \mrf problems, respectively.
There is an instance for which it holds that $\frac{OPT_{\fcfs}}{OPT}=\Omega(n)$.
\end{proposition}
\begin{proof}
Consider an instance consisting of $m$ processors and $n$ jobs, where $m=n$.
The release date of each job $j \in \mathcal{J}$ is $(j-1)\epsilon$, for a very small $\epsilon>0$, and its weight $w_j=1$.
Each job $j \in \mathcal{J}$ consists of $m$ tasks, one per processor.
Moreover, the task $T_{i,j} \in \mathcal{T}$ is a Map task only if $i=j$; otherwise $T_{i,j}$ is a Reduce task.
For each task $T_{i,i} \in \mathcal{M}$, let $v_{i,i}=1$.
For each task $T_{i,j} \in \mathcal{R}$, let $v_{i,j}=\epsilon$.
Let also $E=1$ and $\beta=2$.

Note that, if $\epsilon \ll 1$ then the processing time of each Reduce task can be considered to be very small
in both the optimal schedules for the \mr and the \mrf problems.
So, we can ignore the execution time and the energy consumption of the Reduce tasks.
We only consider the precedence constraints that they imply.

In an optimal solution for the \mr problem, the Map task of job $j$ starts at time $(j-1)\epsilon$.
Due to the convexity and the fact that $w_j=1$ for each $j \in \mathcal{J}$,
we can assume that all Map tasks will be executed with the same speed;
hence the processing time of each Map task is approximately equal to $\sqrt[\beta-1]{\frac{m}{E}}=m$,
as $E=1$ and $\beta=2$.
Thus, the completion time of each job is approximately equal to $m$, and hence $OPT=O(m^2)$.

On the other hand, in an optimal solution for the \mrf problem the Map tasks are not executed in parallel,
as we are forced to respect the order and the precedence constraints.
Ignoring again the processing times of the Reduce tasks, we can assume that the Map task of job $j$ starts at the completion time of job $j-1$.
In order to find the speed $s_j$ of each Map task $T_{j,j} \in \mathcal{T}$
into an optimal solution for the \mrf problem, we have to solve the following convex program.
\begin{equation*}
\text{minimize ~ } \sum_{j=1}^n \frac{n-j+1}{s_j} \text{ ~ subject to ~ } \sum_{j=1}^n s_j \leq E
\end{equation*}
The objective of this convex program corresponds to the objective of the \mrf problem for the given instance,
while the constraint ensures that the selected speeds respect the energy budget.
By applying the Karush-Kuhn-Tucker conditions to this program
we get that $s_j=\frac{E\cdot(n-j+1)^{1/2}}{\sum_{i=1}^n (n-i+1)^{1/2}}$.
By replacing this to the objective we get
\begin{eqnarray*}
OPT_{\fcfs} & = & \sum_{j=1}^n \frac{n-j+1}{\frac{E \cdot (n-j+1)^{1/2}}{\sum_{i=1}^n (n-i+1)^{1/2}}}\\
 & = & \frac{1}{E}\left(\sum_{i=1}^n (n-i+1)^{1/2}\right)^2\\
 & = & \frac{1}{E}\left(\sum_{i=1}^n i^{1/2}\right)^2
 = O\left(\frac{n^3}{E}\right)
\end{eqnarray*}

As $n=m$ and $E=1$, the proposition follows.\qed
\end{proof}

\begin{proposition}\label{prop:sr}
Let $OPT$ and $OPT_{\sr}$ be the optimal solutions for the \mr and the \mrsr problems, respectively.
There is an instance for which it holds that $\frac{OPT_{\sr}}{OPT}=\Omega(n)$.
\end{proposition}
\begin{proof}
We consider a simplified instance which consists of only one processor and does not take into account Map and Reduce tasks and hence precedences.
In this instance the critical issue is the release dates.
For each job $j$, $1 \leq j \leq n-1$, we have $v_j=1$, $w_j=1$ and $r_j=0$,
while for the job $n$ we have $v_n=1-\epsilon$, $w_n=1$ and $r_n=r$, where $r \in \mathbb{R}$ is a big number.
Let $E=1$ and $\beta=2$.

In an optimal schedule for the \mr problem, the jobs $1,2,\ldots,n-1$ are scheduled consecutively starting from time 0,
while the job $n$ is scheduled starting from time $r$.
Let $E_1$ and $E_2$ be parts of the energy budget used for the execution of the jobs $1,2,\ldots,n-1$ and $n$, respectively.
Clearly, it holds that $E_1+E_2=1$.
Hence, following similar analysis as in Proposition~\ref{prop:fcfs} for the \mrf problem,
for the total weighted completion time of the jobs $1,2,\ldots,n-1$ it holds that
\begin{equation*}
\sum_{j=1}^{n-1} w_j C_j = O\left(\frac{n^3}{E_1}\right)
\end{equation*}
The processing time of job $n$ is $E_2$, and hence its completion time is $C_j=r+E_2$.
Therefore, for the optimal solution for the \mr problem we have that
\begin{eqnarray*}
OPT & = & O\left(\frac{n^3}{E_1}\right) + r + \frac{1}{E_2} \\
 & = & O\left(\frac{n^3}{E_1}\right) + r + \frac{1}{1-E_1}
 = r + O(n^3)
\end{eqnarray*}
as this function is minimized for $E_1\simeq1/2$.

On the other hand, in an optimal schedule for the \mrsr problem,
the jobs are scheduled starting from $r$ according to the \sr order, i.e., $\langle n,1,2,\ldots,n-1\rangle$.
As we can choose an $\epsilon$ such that $\epsilon\ll1$, we can assume that all jobs have the same work to execute.
Then, following similar analysis as in Proposition~\ref{prop:fcfs} for the \mrf problem,
we have that $OPT_{SR}=nr+O(n^3)$.

As $r$ can be arbitrary large, the proposition follows. \qed
\end{proof}

\subsection{Experimental Evaluation of Scheduling Policies}

In this section, our goal is to compare the \fcfs and \sr policies with respect to the quality of the solution that they produce.


Our simulations have been performed on a machine with a CPU Intel Xeon X5650 with 8 cores, running at 2.67GHz.
The operating system of the machine is a Linux Debian 6.0.
We used Matlab with cvx toolbox.
The solver used for the convex program is SeDuMi.

The instance of the problem consists of a matrix $m \times n$ that corresponds to the work of the tasks,
two vectors of size $n$ that correspond to the weights and the release dates of jobs,
a precedence graph for the tasks of the same job,
the energy budget and the value of $\beta$.

Similarly with \cite{ChenKL12}, the instance consists of $m=50$ processors and up to $n=25$ jobs.
Each job has 20 Map and 10 Reduce tasks, which are preassigned at random to a different processor.
The work of each Map task is selected uniformly at random in $[1,10]$,
while the work of each Reduce task $v_{i,j} \in \mathcal{R}$ is equal to a random number in $[1,10]$ plus
$\frac{3\sum_{T_{i',j} \in \mathcal{M}} v_{i',j}}{|\{T_{i',j} \in \mathcal{M}\}|}$,
taking into account the fact that Reduce tasks have more work to execute than Map tasks.
The weight of each job is selected uniformly at random in $[1,10]$.
For the release date of a job, we select with probability 1/2 every interval $(t,t+1]$.
Then, the release date is equal to a random value in this interval.
The energy budget that we used is $E=1000$.
We have also set $\beta=2$.
%
We set the desired accuracy of the returned solution of the convex program to be equal to $10^{-7}$.
For each number of jobs we have repeated the experiments with 10 different matrices.
The results we present below, concern the average of these 10 instances.

The benchmark as well as the code we used in our experiments are freely available at\\
\url{http://www.ibisc.univ-evry.fr/~vchau/research/mapreduce/}.


As mentioned before, the (CP) does not lead to a feasible solution for our problem.
In order to get such a solution we apply the following algorithm.
At each time $t$ where a processor becomes available we select to schedule the task $T_{i,j}$ of higher priority such that:
(i) $T_{i,j}$ is already released at $t$,
(ii) if $T_{i,j}$ is a Reduce task, then all Map tasks of the same job have been already completed at $t$, and
(iii) $T_{i,j}$ has not been yet executed.

As shown in Fig.~\ref{fig:exp} the heuristic based on \fcfs outperforms the heuristic based on \sr.
In fact, the first heuristic gives up to $16-21\%$ better solutions that the second one for different values of $n$.
Surprisingly, the situation is completely inverse if we consider the corresponding solutions of the convex programs.
More precisely, the convex programming relaxation using \sr leads to $26\%-43\%$ smaller values of the objective function
with respect to the convex programming relaxation using \fcfs.

Moreover, we can observe that the ratio between the final solution of each heuristic with respect to the lower bound for the \mrs problem
given by the convex program is equal to 1.46 for \fcfs and 2.43 for \sr;
the variance is less than 0.1 in both cases.
However, as we already mentioned, this ratio cannot be considered as the approximation ratio for the \mr problem,
as its optimal solution can be significantly smaller than the optimal solution for the \mrs problem using the \fcfs and \sr orders.

\section{Conclusions}\label{se:con}

We presented a constant-approximation algorithm for the problem of scheduling a set of MapReduce jobs
in order to minimize their total weighted completion time under a given budget of energy.
Our algorithm uses an optimal solution to an LP relaxation in interval-indexed variables and converts it to a feasible
non-preemptive schedule of the \mr problem using the idea of list scheduling in order of a-points.
Moreover, we proposed a convex programming relaxation of the problem when a prespecified order of jobs is given.
Based on the solution of this convex programming relaxation, we explored the efficiency of standard scheduling policies,
by presenting counterexamples for them as well as by experimentally evaluating their performance.
It has to be noticed that our results can be extended also to the case where multiple Map or Reduce tasks of a job are executed on the same processor.
An interesting direction for future work concerns the online case of the problem.
Although, it can be proved that there is no an $O(1)$-competitive deterministic algorithm (see Theorem~13 in~\cite{BansalPS09}),
a possible way to overcome this  is to consider resource (energy) augmentation,
or to study the closely-related objective of a linear combination of the sum of weighted completion times of the jobs and of the total consumed energy.

%

\bibliographystyle{plain}

%
%
%
%
%
%
%
%
%
%
%
%
%
%
%
%
%
%
%
%
%

\end{document}